\documentclass[journal]{IEEEtran}
\usepackage{amsmath}
\usepackage{amssymb}
\usepackage{amsthm}
\usepackage{mathrsfs}
\usepackage{mathtools}
\usepackage{color}
\usepackage{array}
\usepackage{multirow}
\usepackage{xcolor}
\usepackage{cite}

\usepackage{etoolbox}
\makeatletter
\patchcmd{\@makecaption}
{\scshape}
{}
{}
{}
\makeatletter
\patchcmd{\@makecaption}
{\\}
{.\ }
{}
{}
\makeatother

\usepackage{flushend}
\usepackage{comment}
\usepackage{booktabs, multirow} 
\usepackage{url}            

\newcommand{\R}{\mathbb{R}}

\renewcommand{\epsilon}{\varepsilon}

\newcommand{\C}{\mathcal{C}}

\renewcommand{\dim}{n}

\newcommand{\X}{\mathcal{X}}

\newcommand{\bit}{\{0,1\}}
\newcommand{\bits}{\bit^{\dim}}

\newcommand{\BI}{\ensuremath{\mathsf{Ball}_\mathsf{I}}}
\newcommand{\BD}{\ensuremath{\mathsf{Ball}_\mathsf{D}}}

\newcommand{\Ball}{\mathsf{Ball}}

\newcommand{\VH}{\ensuremath{V_\mathsf{H}}}
\newcommand{\VI}{\ensuremath{V_\mathsf{I}}}
\newcommand{\VD}{\ensuremath{V_\mathsf{D}}}

\newcommand{\len}[1]{\mathsf{len}(#1)}
\newcommand{\run}[1]{\rho(#1)}

\newcommand{\KI}{\ensuremath{\mathsf{K}_\mathsf{I}}}
\newcommand{\KD}{\ensuremath{\mathsf{K}_\mathsf{D}}}

\newtheorem{defn}{Definition}
\newtheorem{thm}{Theorem}
\newtheorem{lemma}[thm]{Lemma}
\newtheorem{corollary}[thm]{Corollary}

\newtheorem{remark}{Remark}
\newtheorem{proposition}[thm]{Proposition}

\newcolumntype{C}[1]{>{\centering\let\newline\\\arraybackslash\hspace{0pt}}m{#1}}

\newcommand{\bfa}{{\boldsymbol a}}
\newcommand{\bfb}{{\boldsymbol b}}
\newcommand{\bfc}{{\boldsymbol c}}

\newcommand{\bfs}{{\boldsymbol s}}

\newcommand{\bfx}{{\boldsymbol x}}
\newcommand{\bfy}{{\boldsymbol y}}

\begin{document}

\title{Covering Codes using Insertions or Deletions}

\author{Andreas~Lenz,
        Cyrus~Rashtchian,
        Paul~H.~Siegel,
        and~Eitan~Yaakobi
\thanks{A. Lenz is with the Institute for Communications Engineering, Technische Universit\"at M\"unchen, Munich 80333, Germany (e-mail: andreas.lenz@mytum.de).}%
\thanks{C. Rashtchian is with the Computer Science and Engineering Department and the Qualcomm Institute, University of California, San Diego, La Jolla, CA 92093-0407 USA (e-mail:	crashtchian@eng.ucsd.edu).}
\thanks{P.  H.  Siegel  is  with  the  Electrical  and  Computer  Engineering  Department and the Center for Memory and Recording Research, University of California, San Diego, La Jolla, CA 92093-0407 USA (e-mail:	psiegel@ucsd.edu).}%
\thanks{E. Yaakobi is with the Computer Science Department, Technion -- Israel Institute of	Technology,	Haifa 3200003, Israel	(e-mail: yaakobi@cs.technion.ac.il).}%
\thanks{A. Lenz acknowledges support from the German-American Fulbright Commission for funding the visit to UCSD. E. Yaakobi acknowledges support from the Center for Memory and Recording Research at UCSD. This work is also funded by the European Research Council under the EU’s Horizon 2020 research and innovation programme (grant No. 801434), by NSF grant CCF-BSF-1619053 and by the United States-Israel BSF grant 2015816.
}%
}
\maketitle

\begin{abstract}
A covering code is a set of codewords with the property that the union of balls, suitably defined, around these codewords covers an entire space. Generally, the goal is to find the covering code with the minimum size codebook. 
While most prior work on covering codes has focused on the Hamming metric, we consider the problem of designing covering codes defined in terms of either insertions or deletions. First, we provide new sphere-covering lower bounds on the minimum possible size of such codes. Then, we provide new existential upper bounds
on the size of optimal covering codes for a single insertion or a single deletion that are tight up to a constant factor.
Finally, we derive improved upper bounds for covering codes using $R\geq 2$ insertions or deletions. We prove that codes exist with density that is only a factor $O(R \log R)$ larger than the lower bounds for all fixed~$R$. In particular, our upper bounds have an optimal dependence on the word length, and we achieve asymptotic density matching the best known bounds for Hamming distance covering codes. 
\end{abstract}

\IEEEpeerreviewmaketitle

\section{Introduction}
 
	Covering codes are a core object of study in coding theory and discrete mathematics. They have found applications in diverse areas such as data compression~\cite{cohen1997covering},  football pools~\cite{hamalainen1995football}, circuit complexity~\cite{smolensky1993representations},  
	lattice problems~\cite{micciancio2004almost}, and approximate nearest neighbor search~\cite{coverlsh1}. Previous work has mostly studied covering codes with respect to substitutions (i.e., the Hamming distance). Recently, due to the large amount of textual and biological data, there has been a resurgence of interest in the Levenshtein distance and in channels with insertion and deletion errors (e.g.,~\cite{braverman2017coding, chakraborty2018approximating, cheraghchi2019overview, mitzenmacher2009survey, levenshtein2001efficient, liron2015characterization, yazdi_portable_2017, organick2018random, chandak_improved_2019}). Despite this substantial progress, the Levenshtein distance remains poorly understood compared to other metrics on discrete spaces, and many fundamental questions remain open.
	
	In this paper, we study covering codes for either insertions or deletions. Loosely speaking, we aim to cover a space of words by the union of balls around a minimum number of codewords. Let $\Sigma_q^n$ denote the set of words of length $n$ over a $q$-ary alphabet. For the case of insertions, a codeword $\bfc \in \Sigma_q^n$ {\em covers} a word $\bfy \in \Sigma^{n+R}$ at radius $R$ if $\bfy$ can be obtained from~$\bfc$ by inserting exactly $R$ symbols. Similarly, for the case of deletions, a codeword $\bfc \in \Sigma_q^n$ {\em covers} a word $\bfy \in  \Sigma^{n-R}$ at radius $R$ if $\bfy$ can be obtained from~$\bfc$ by deleting exactly $R$ symbols. This means that the covering problem for insertions deals with finding a small set of codewords of length $n$ such that each word of length $n+R$ is a supersequence of a codeword. Analogously, for the case of deletions, each word of length $n-R$ must be a subsequence of some codeword. In both cases, balls are naturally defined by the set of words obtained by inserting $R$ symbols into or deleting $R$ symbols from some word of length $n$. Notice, however, that the codewords and the covered words reside in different spaces because they have different lengths. Hence, the covering problem for insertions and deletions is inherently asymmetric.

	Although there is a rich literature on covering codes for the Hamming distance~\cite{cohen1997covering}, as well as recent improvements for insertion/deletion error-correcting codes (e.g.,~\cite{bukh2016improved, guruswami2019optimally, guruswami2018polynomial, guruswami2017deletion, haeupler2019near, sima_optimal_2019}), much less is known about covering codes using insertions or deletions. Two key challenges are the (ir)regularity of the balls and the asymmetry of the covering problem. Insertion balls are regular, in the sense that for any $\bfx\in \Sigma_q^n$ and $R \geq 1$, there are exactly $\sum_{i=0}^{\dim +R} \binom{\dim+R}{i}(q-1)^i$ words of length $\dim + R$ obtainable by inserting $R$ symbols into $\bfx$ (cf.~\cite{Levenshtein1974}). In contrast, deletion balls are irregular, and their sizes depend on many properties of their center, such as the number of runs. In fact, a tractable exact formula remains unknown for the size of the deletion balls with radius three or greater. This irregularity and lack of an explicit formula for the ball size  mean that, compared to\hspace{-.04cm} the\hspace{-.04cm} Hamming\hspace{-.04cm} distance,\hspace{-.04cm} it\hspace{-.04cm} is\hspace{-.04cm} inherently\hspace{-.04cm} more\hspace{-.04cm} challenging to derive bounds on the minimum covering code size, even asymptotically.
	
	In some cases, we can infer results on covering codes from the theory of error-correcting codes. This is due to the existence of \emph{perfect} error-correcting codes, for which the balls of radius $R$ around all codewords are not only distinct but also cover each word once. For example, the Varshamov-Tenengolts (VT) code is a perfect binary single-deletion-correcting code~\cite{Levenshtein1992}. It is known that the VT code is the largest single-deletion-correcting code for $n\leq 14$~\cite{Fazeli2015}, and this is conjectured to be true for $n > 14$ (see Sloane~\cite[Conj. 2.6]{Sloane2002}). This conjecture however remains open. Nevertheless, since the VT code is indeed a perfect single-deletion-correcting code, it is also a single-deletion-covering code. 
	
	\newcommand{\ds}{\displaystyle}
	\begin{table*}[t]
		\centering
		\caption{Upper and lower bounds for covering codes $\C \subseteq \Sigma_q^n$ using substitutions, insertions, and deletions. We let $c$ denote a universal constant. We denote the size of a radius-$R$ Hamming ball by $\VH^q(\dim,R) = \sum_{i=0}^R \binom{\dim}{i}(q-1)^i$, and the size of a radius-$R$ insertion ball by $\VI^q(n,R)= \sum_{i=0}^R \binom{\dim+R}{i}(q-1)^i$.  Entries marked with ``${\color{darkgray} (\infty)}$'' are asymptotic results for fixed $R$ and large $n$, where a factor of $1\pm o(1)$ has been omitted for readability.\vspace{-1em}}
		\setlength{\tabcolsep}{18pt}
		{\renewcommand{\arraystretch}{3.2}
			\begin{tabular}{lc@{\hskip 5pt}cc@{\hskip 5pt}cl}
				\specialrule{.7pt}{11pt}{-8pt}
				Covering Code Type & Existence Size &  & Lower Bound & & Reference \\[-5pt]
				\specialrule{.6pt}{1pt}{-1pt}
				1-substitution & $\ds \frac{q^\dim}{(q-1)\dim+1}$ & $^{\color{darkgray} (\infty)}$ & $\ds \frac{q^\dim}{(q-1)\dim+1}$&&  \cite{kaba88}\\
				$R$-substitution & $\ds  \frac{cR\log R \cdot q^n}{\VH^q(\dim, R)}$& $^{\color{darkgray} (\infty)}$ & $\ds \frac{q^n}{\VH^q(\dim, R)}$ && \cite{krivelevich2003covering}\\
				%
				\specialrule{.5pt}{5pt}{-1pt}
				1-insertion & $\ds \frac{7 \cdot q^{n+1}}{(n+1)(q-1)+1}$ && $\ds \frac{q^{n+1}}{(n+1)(q-1)+1}$ && Theorem \ref{thm:one:insertion}, Theorem \ref{thm:lower:bound:insertion} \\
				$R$-insertion & $\ds \frac{cR\log R \cdot q^{n+R}}{\VI^q(n,R)}$ & $^{\color{darkgray} (\infty)}$ & $\ds \frac{q^{n+R}}{\VI^q(n,R)}$ && Theorem \ref{thm:k:insertion}, Theorem \ref{thm:lower:bound:insertion} \\
				%
				\specialrule{.5pt}{5pt}{-1pt}
				1-deletion (binary) & $\ds \frac{2^n}{\dim+1}$ && $\ds \frac{2^n}{\dim+1}$& $^{\color{darkgray} (\infty)}$& \cite{afrati-anchor, VT65} \\
				1-deletion & $\ds \frac{q^{\dim}}{(\dim+1)\lfloor q/2\rfloor}$ && $\ds \frac{q^{\dim}(n-2)}{(q-1)n(n+1)}$ && Theorem \ref{thm:single:q:deletion}, Theorem \ref{thm:lower:bound:deletions} \\
				$R$-deletion & $\ds \frac{cR\log R \cdot q^{n}R!}{n^R(q-1)^R}$ & $^{\color{darkgray} (\infty)}$ & $\ds \frac{q^nR!}{n^R(q-1)^R} $& $^{\color{darkgray} (\infty)}$ & Theorem \ref{thm:k:deletion}, Theorem \ref{thm:lower:bound:deletions} \\
				\bottomrule
		\end{tabular}}
		\label{table:code-results}
	\end{table*}
	
	While it has been shown that an $R$-deletion-correcting code is equivalent to an $R$-insertion-and-deletion-correcting code \cite{Lev65}, this property does not hold for the case of covering codes. This means that the VT codes are not single-insertion-covering codes and thus also not perfect codes for correcting a single insertion. In fact, it has been shown that the only perfect single-insertion-correcting codes are binary and have length two~\cite{Levenshtein1992}. Therefore, the best possible size of a single-insertion-covering code is unknown, and constructing optimal covering codes in this case is a highly non-trivial problem, which we address in this paper.

Afrati et al. have studied covering codes for insertions and deletions, motivated by designing MapReduce algorithms for similarity joins under the Levenshtein distance~\cite{afrati-anchor,afrati-fuzzy}. They show the existence of single-insertion-covering codes with size $O(\frac{q^\dim \log \dim}{\dim})$, while they prove a lower bound stating that such codes must have at least $\frac{q^\dim}{(q-1)\dim+1}$ codewords.  
As one of our contributions, we certify that the lower bound is nearly tight by showing the existence of single-insertion-covering codes with $\Theta(\frac{q^\dim}{\dim})$ codewords. Afrati et al. also provide explicit constructions, albeit using more codewords~\cite{afrati-anchor}. They construct single-insertion-covering codes of size $O(q^{n-2})$ and double-insertion-covering codes of size $O(q^{n-3})$. Afrati et al. also provide an explicit construction of a single-deletion-covering code with $q^n/n$ codewords.
Finally, we note that tensorization arguments may be used to build radius-$R$ covering codes from radius-one codes~\cite{cohen1997covering}, but this approach introduces a factor of $R^{O(R)}$ with respect to the sphere covering lower bound, leading to larger covering codes than desired. 

We mention that covering codes for insertions or deletions are similar in spirit to asymmetric covering codes for substitutions~\cite{cooper2002asymmetric}. There, the goal is to find a code $\C \subseteq \{0,1\}^n$ such that the Hamming distance of any word $\bfy \in \{0,1\}^n$ is at most $R$ from some $\bfc \in \C$, while satisfying $\bfy \leq \bfc$ in the binary partial order. Asymmetric covering codes are related to the Erd\'{o}s-Hanani conjecture on hypergraph coverings, which has been resolved affirmatively by R\"{o}dl~\cite{hanani1963limit, rodl1985packing}. For more information on covering codes, we refer the reader to the  book by Cohen, Honkala, Litsyn, and Lobstein~\cite{cohen1997covering}.

\subsection{Our Results}

	We provide new upper and lower bounds on the minimum size of insertion-covering and deletion-covering codes. We primarily consider the size of such codes for fixed alphabet size $q$ and covering radius $R$.
	Table~\ref{table:code-results} summarizes our results.
	The bounds are stated separately for $R=1$ and general $R\geq1$ because we obtain tighter bounds in the former case. The first two rows of Table~\ref{table:code-results} also recap the best known bounds for substitution-covering codes.\footnote{We often refer to Hamming distance covering codes as $R$-substitution-covering codes for consistency.}
	
	For $R$-insertion-covering codes, the goal is to cover words of length $\dim + R$ using words of length~$n$. In the case of $R=1$, we provide nearly matching upper and lower bounds that differ only by a factor of seven. This improves upon the upper bound result of Afrati et al.~\cite{afrati-anchor} by a $\Theta(\log \dim)$ factor. For $R\geq 2$ insertions, we prove that $R$-insertion-covering codes exist with size that is off by a factor of $O(R \log R)$ from the lower bound (the dependence on the dimension $n$ and alphabet size $q$ are optimal). We remark that the gap between upper and lower bounds matches the state-of-the-art for $R$-substitution-covering codes~\cite{krivelevich2003covering}, and it seems beyond our current techniques to obtain a tighter bound.
	
	For the case of $R$-deletion-covering codes, the goal is to cover words of length $n-R$ by deleting $R$ symbols from codewords of length $n$. First, we provide a new lower bound on the minimum size of $R$-deletion-covering codes. Then, for words over a $q$-ary alphabet with $q >2$, we provide a new explicit construction of single-deletion-covering codes, where the number of codewords is within a factor of two from optimal. 
	Finally, for $R \geq 2$ deletions, we prove that $R$-deletion-covering codes exist with size that is tight up to a factor of $O(R \log R)$ compared to the lower bound.

	We note that our upper bounds for $R$ insertions (resp. $R$ deletions) will depend upon the size of covering codes for a single insertion (resp. single deletion). In particular, establishing a better upper bound for a single insertion/deletion would immediately lead to smaller codes for radii $R>1$.
	
	Perhaps surprisingly, our existential bounds for insertion and deletion covering codes with radii $R \geq 2$ have the same quantitative overhead as the best known bounds for substitution-covering codes. One reason for this apparent similarity is that our proof techniques are inspired by arguments for substitutions. We observe that the previous arguments are not specific to Hamming distance, but rather, they hold for any space on words that (i) satisfies mild conditions on the distribution of the respective covering ball sizes and (ii) allows a semi-direct sum operation, where covering codes over shorter words may be combined to assemble codes over longer words. As both insertion balls and deletion balls satisfy these two criteria, the same overall proof strategy goes through, even though there are several technical differences. More generally, we believe that this proof technique may be useful for other discrete geometric spaces as well.
	
	While in this work we analyze covering codes using either only insertions or deletions, it would be interesting to study combinations of insertions and deletions as well. We largely leave this as future work, briefly noting certain challenges that arise. Defining the covering problem already leads to flexibility because the length of the covered words may now vary, instead of being fixed. For example, if we wished to use at most $R_I$ insertions and at most $R_D$ deletions, then we could ask: what is the minimum number of length-$n$ codewords that suffice to cover all words with lengths between $n-R_D$ and $n + R_I$? Our methods may extend more easily to formulations where the number of insertions and deletions are fixed individually. For instance, we could consider using exactly $R_I$ insertions and $R_D$ deletions to cover words of length $n+R_I-R_D$. This is in fact an interesting extension of our work, since even when $R_I = R_D = R_S/2$, it is already difficult to understand how this covering problem compares to $R_S$-substitution-covering codes over words of length exactly $n$.
	In general, obtaining nearly-tight upper and lower bounds for these various covering problems would require accurate estimates of the number of words obtained from combinations of insertions and deletions, which is an active area of research that is studied in different contexts as well~\cite{bukh2019periodic,ganguly2016sequence, kiwi2005expected, lueker2009improved, navarro2001guided, shmid2019bounds}.

	\subsection{Overview of Our Techniques}
	
	A lower bound on the necessary size of $R$-insertion-covering codes can be derived from known bounds on the number of words obtained by $R$ insertions. For example, in the case of a single insertion, each codeword of length $n$ covers $(n+1)(q-1)+1$ distinct words. Thus, one would hope for a code of size $O(q^n/n)$. A natural approach would be to adapt or modify known codes, such as VT codes or Hamming codes.  Afrati et al.~\cite{afrati-anchor} take this approach. After converting $q$-ary words to binary, they use $O(\log \dim)$ translates of a Hamming code as the basis for their single-insertion-covering code with $O(q^\dim \log \dim / \dim)$ codewords.\footnote{In fact, it is also possible to achieve the same asymptotic result by analyzing a greedy algorithm that successively adds codewords based on how many new words they cover. This algorithm yields a covering code and the Johnson-Stein-Lov\'{a}sz theorem \cite{cohen1997covering} gives an upper bound of $O(q^n\log(n)/n)$ for the size of single-insertion-covering codes.} To remove the $\log n$ factor, we use a combination of a random construction with a careful inductive argument, inspired by the proof of Cooper, Ellis, and Kahng on asymmetric covering codes~\cite{cooper2002asymmetric}. We further refine their technique to obtain explicit values on the resulting code sizes. Our upper bound for $R \geq 2$ insertions uses a generalization of the previous argument. In particular, we will roughly follow the high-level strategy used by Krivelevich, Sudakov, and Vu to obtain the current best bounds on  $R$-substitution-covering codes~\cite{krivelevich2003covering}.   The main idea is to first cover a large fraction of words by randomly sampling a subset of possible codewords (where codewords are included with probabilities depending on certain deletion ball sizes). Then, we use  a direct-sum-type operation to cover the remaining words (where this operation recursively utilizes covering codes for smaller word lengths or covering radii).

	Turning to deletions, we derive a precise lower bound on the minimum size of deletion-covering codes. Obtaining such a bound is somewhat involved because the sizes of deletion balls are different even for words of the same length. Hence, standard sphere-covering arguments do not apply. We use a technique, due to Applegate, Reins, and Sloane~\cite{Applegate03}, which enables a covering lower bound, even though the sizes of the balls are non-uniform. This technique uses an integer programming approach to analyze a weighted covering. To apply this, we construct a non-uniform weight function, depending on the sizes of deletion balls. This approach is related to bounding the size of deletion-correcting codes, which requires a generalized sphere-packing bound~\cite{Fazeli2015, Kulkarni13}. For our upper bounds, we have already noted that VT codes provide asymptotically optimal single-deletion-covering codes for binary words. Unfortunately, known non-binary single-deletion-correcting codes \cite{T84} are not perfect, and thus, it is non-trivial to obtain good covering codes for this case. We provide a new generalization of VT codes to non-binary alphabets, and show that this leads to an explicit construction of nearly optimal covering codes for a single deletion. For $R \geq 2$ deletions, we use an analogous argument as for insertion-covering codes, combining random sampling with a recursive construction. 
	
	The rest of the paper is organized as follows. Section~\ref{sec:def} presents the notations and the definitions of covering codes for insertions and deletions. Section~\ref{sec:bounds} contains lower bounds on the cardinality of insertion- and deletion-covering codes. Section~\ref{sec:single_del_nb} is dedicated to the case of single-deletion-covering codes, and Section \ref{sec:single} presents existence bounds for single-insertion-covering codes. In Section~\ref{sec:multiple}, we extend the results to multiple insertions and multiple deletions. Lastly, Section~\ref{sec:conc} concludes the paper and discusses open problems.

\section{Notations, Definitions, and Preliminaries}\label{sec:def}
\label{sec:notation}

For an integer $q\geq 2$, let $\Sigma_q$ denote the $q$-ary alphabet $\{0,1,\ldots,q-1\}$  and $\Sigma_q^* = \bigcup_{\ell \geq 0} \Sigma_q^\ell $. We use $\len{\bfx}$ to denote the length of $\bfx$. 
For $\bfx = (x_1,\ldots, x_\dim) \in \Sigma_q^n$, we let $\run{\bfx}$ denote the number of runs in $\bfx$, that is,
\[\run{\bfx} := 1 + |\{1 \leq i < \dim : x_i \neq x_{i+1}\}|.\]
For $\bfx,\bfy \in \Sigma_q^*$, the notation $\bfx\bfy$ denotes the concatenation of $\bfx$ and $\bfy$, where $\len{\bfx\bfy} = \len{\bfx}+\len{\bfy}$.

%
For $\bfx\in\Sigma_q^n$, we abbreviate the {\bf radius-$t$ insertion ball} obtained after exactly $t$ insertions by $\BI^q(\bfx,t)$ and its size is denoted by $\VI^q(\bfx,t)$. Similarly, the {\bf radius-$t$ deletion ball} obtained after exactly $t$ deletions is denoted by $\BD^q(\bfx,t)$ and its size is $\VD^q(\bfx,t)$. It is well known, see e.g.~\cite{Levenshtein1974}, that while the size of a deletion ball depends heavily on its center $\bfx$, insertion balls are regular. Thus, we denote by $\VI^q(n,t)$ the insertion ball size of length-$n$ words over $\Sigma_q^n$. 
	
	We will consider two sub-problems, namely covering words with only insertions or only deletions. We begin with insertions, where codewords have length $\dim$ and they cover words of length $\dim +R$ after $R$ insertions. 	We seek  codes in $\Sigma_q^n$ of minimum cardinality such that the union of all radius-$R$ insertion balls around codewords contains the whole hypercube $\Sigma_q^{\dim+R}$. More formally, we have the following.
\begin{defn}
A code $\C \subseteq \Sigma_q^n$ is an {\bf $R$-insertion-covering code}, if for every  $\bfy \in \Sigma_q^{\dim+R}$, there exists a codeword $\bfc \in \C$ such that $\bfy \in \BI^q(\bfc, R)$. That is, $\bigcup_{\bfc\in\C}\BI^q(\bfc, R) =  \Sigma_q^{\dim+R}$.
\end{defn}
	
An $R$-deletion-covering code is defined similarly, but codewords cover words of length $\dim-R$ after $R$ deletions.
\begin{defn}
A code $\C \subseteq \Sigma_q^n$ is an {\bf $R$-deletion-covering code}, if for every $\bfy \in \Sigma_q^{\dim-R}$, there exists a codeword $\bfc \in \C$ such that $\bfy \in \BD^q(\bfc, R)$.
That is, $\bigcup_{\bfc\in\C}\BD^q(\bfc, R) = \Sigma_q^{\dim-R}$.
\end{defn}

The {\bf insertion (resp. deletion)  radius} of a code $\C$ is defined to be the smallest $R$ such that $\C$ is an $R$-insertion-covering (resp. $R$-deletion-covering) code. We also denote by $\KI^q(n,R)$  (resp. $\KD^q(n,R)$)   the smallest cardinality of an $R$-insertion-covering (resp. $R$-deletion-covering) code, of length $n$ over~$\Sigma_q$.  When discussing the binary case, i.e., $q=2$, we will typically remove $q$ from the above notations. 

\section{Lower Bounds}\label{sec:bounds}
	In this section, we establish lower bounds on the size of insertion- and deletion-covering codes based on a sphere covering argument. As in the case of substitution-covering codes, the argument relies on the union bound and an upper bound on the number of words a codeword can cover.   For the case of insertions, the insertion-ball size is known and independent of the ball center.  The case of deletions is more challenging due to the dependence  of the deletion-ball size on the ball center.

We start with the easier case of insertions. 
	\begin{thm} \label{thm:lower:bound:insertion}
	For all $n$ and $R$, it holds that 
	\begin{equation}\label{eq:insertion:lower:exact}
	\KI^q(n,R) \geq \frac{q^{\dim+R}}{\VI^q(n,R)}   =  \frac{q^{\dim+R}}{\sum_{i=0}^{R} \binom{n+R}{i} (q-1)^i}.
	\end{equation}
	Furthermore, for fixed $R$ and large $n$,
	\begin{equation}\label{eq:insertion:lower:approx}
	\KI^q(n,R) \geq \frac{R!q^{\dim+R}}{n^R(q-1)^R}(1-o(1)).
	\end{equation} 
	\end{thm}
	\begin{proof}
	Let $\C$ be any $R$-insertion-covering code. Hence, we have that $\bigcup_{\bfc \in \C} \BI^q(\bfc,R) = \Sigma_q^{n+R}$. Computing the cardinalities of both sets, we obtain
		\begin{align*}
		q^{n+R} & = |\Sigma_q^{n+R}| = \Big|\bigcup_{\bfc \in \C} \BI^q(\bfc,R)\Big| & \\
		& \overset{(a)}{\leq} \sum_{\bfc \in \C} \VI^q(\bfc,R) = |\C| \sum_{i=0}^{R} \binom{n+R}{i} (q-1)^i, 
		\end{align*}
		where we used the union bound in inequality $(a)$ and the fact that for any $\bfx\in\Sigma_q^n$ the size of the insertion balls only depends on the length of $\bfx$ and is given by $\VI^q(n,R) = \sum_{i=0}^{R} \binom{n+R}{i}(q-1)^i$ \cite{Levenshtein1974}. Reading the  above inequality from right to left, we obtain the bound (\ref{eq:insertion:lower:exact}). 
		
		The approximation (\ref{eq:insertion:lower:approx}) for large $\dim$ is obtained by the following standard inequality, 
		$$\frac{q^{\dim+R}}{\sum_{i=0}^{R} \binom{\dim+R}{i}(q-1)^i}
		\geq 
		\frac{R!q^{\dim+R}}{n^R(q-1)^R}(1-o(1)),$$
which is proved in Proposition~\ref{prop:insertion:approx:lb} in the appendix.
	\end{proof}

	For the case of deletions, deriving a sphere-covering lower bound is  more involved due to the fact that the size of the deletion ball  $\BD(\bfx,R)$ can be different for words of the same length. To overcome this difficulty, we use a  technique due to Applegate et al. \cite{Applegate03} that enables the computation of a bound even though the ball sizes are irregular. We restate the lemma from \cite{Applegate03} in a form suited to our particular context.
	
	\begin{lemma}[cf. \cite{Applegate03}] \label{lemma:generalized:sphere:covering:bound}
	Let $U\subseteq \Sigma_q^*$, $T \subseteq \Sigma_q^*$ be arbitrary sets and $\Ball: U \mapsto \{t:t\subseteq T\}$ be an arbitrary mapping. Further let $w: T \to \mathbb{R}$ be a weight function satisfying
		\begin{equation} \label{eq:weighting:condition}
		\sum_{\bfy \in\Ball(\bfx)} w(\bfy) \leq 1
		\end{equation}
		for all $\bfx \in U$. Then any covering code $\mathcal{C} \subseteq U$, which covers $T$, i.e., $T = \bigcup_{\bfc \in \C}\Ball(\bfc)$, satisfies
		$$ |\mathcal{C}| \geq \sum_{\bfy \in T}w(\bfy). $$
	\end{lemma}
\noindent
Even though we will derive the sphere-covering lower bound for deletions only, the above statement holds in general for any $\mathsf{Ball}(\bfx)$, as has been proven in \cite{Applegate03}. We note that for error-correcting codes, in which the space is replaced by $\Sigma_q^*$, the analogous bound is called a generalized sphere-packing bound and has been studied in~\cite{Fazeli2015, Kulkarni13}. 

	The maximum  possible sum weight $w(\bfy)$ that fulfills \eqref{eq:weighting:condition} offers the best lower bound on the size of an $R$-deletion-covering code.  Finding it requires solving a linear programming problem, as specified in Lemma~\ref{lemma:generalized:sphere:covering:bound}. Here, we choose the weight function  to approximate the inverse of the deletion-ball size,  $w(\bfy) \approx \VD^q(\bfx,R)^{-1}$, where we recall that  $\VD^q(\bfx,R) = |\BD^q(\bfx,R)|$. Under the assumption that for $\bfy \in \BD^q(\bfx,R)$ the deletion  balls have approximately the same size,  $\VD^q(\bfx,R)\approx\VD^q(\bfy,R)$, condition \eqref{eq:weighting:condition} is fulfilled with sum-weight  close to $1$. We will provide a rigorous derivation based upon this intuition in the following theorem.

	\begin{thm} \label{thm:lower:bound:deletions}
	For all $n$ and $0< R < n$, it holds that 
		$$ \KD^q(n,R) \geq q\sum_{r=1}^{\dim-R} \frac{(q-1)^{r-1}\binom{n-R-1}{r-1}}{\binom{r+3R-1}{R}}.$$
		In particular, for $R=1$ we get that 
		$$ \KD^q(n,1) \geq \frac{q^{\dim}(n-2)}{(q-1)n(n+1)}.$$
		Furthermore, for fixed $R$ and large $n$, we have
		$$ \KD^q(n,R) \geq \frac{R!q^{n}}{n^R(q-1)^{R}}  \left( 1 -o(1)  \right). $$
	\end{thm}
	\begin{proof}
		We will prove the theorem using Lemma \ref{lemma:generalized:sphere:covering:bound}. We take $U = \Sigma_q^n$ and $T = \Sigma_q^{\dim-R}$. Define the weight function $w: T \to \mathbb{R}$ to be 
		$$w(\bfy) = \frac{1}{\underset{\bfx: \bfy \in \BD^q(\bfx,R)}{\max}\VD^q(\bfx,R)}, $$
		for all $\bfy \in \Sigma_q^{\dim-R}$. It follows that for all $\bfy \in \BD^q(\bfx,R)$, $w(\bfy) \leq  \frac{1}{\VD^q(\bfx,R)}$, which implies that $w$ satisfies \eqref{eq:weighting:condition}, since
		$$ \sum_{\bfy \in \BD^q(\bfx,R)}w(\bfy) \leq \sum_{\bfy \in \BD^q(\bfx,R)} \frac{1}{\VD^q(\bfx,R)} = 1 $$
		for all $\bfx \in \Sigma_q^n$. By Lemma~\ref{lemma:generalized:sphere:covering:bound} we obtain a lower bound on the cardinality of any $R$-deletion-covering code $\C$, as follows.
		\begin{align*}
			|\C| & \geq \sum_{\bfy \in \Sigma_q^{\dim-R}}w(\bfy) = \sum_{\bfy \in \Sigma_q^{\dim-R}}\frac{1}{\underset{\bfx: \bfy \in \BD^q(\bfx,R)}{\max}\VD^q(\bfx,R)} & \\
			&  \overset{(a)}{\geq} \sum_{\bfy \in \Sigma_q^{\dim-R}} \frac{1}{\underset{\bfx: \bfy \in \BD^q(\bfx,R)}{\max}\binom{\run{\bfx}+R-1}{R}}, &
		\end{align*}
		where in $(a)$ we used the fact that the size of the radius-$R$ deletion ball is at most $\VD^q(\bfx,R) \leq \binom{\run{\bfx}+R-1}{R}$~\cite{Lev65}. Moreover, for all $\bfy \in \BD^q(\bfx,R)$ we have that $\run{\bfx} \leq \run{\bfy}+2R$, since each deletion can eliminate at most two runs. Thus,
		$$|\C| \geq \sum_{\bfy \in \Sigma_q^{\dim-R}} \frac{1}{\binom{\run{\bfy}+3R-1}{R}} \overset{(b)}{=} q\sum_{r=1}^{\dim-R} \frac{(q-1)^{r-1}\binom{n-R-1}{r-1}}{\binom{r+3R-1}{R}},$$
		where in $(b)$ we used the fact that the number of words $\bfx \in \Sigma_q^{\dim}$ with $\run{\bfx} = r$ runs is given by $q(q-1)^{r-1} \binom{\dim-1}{r-1}$~\cite{Lev65}, and we combined terms corresponding to $\bfy\in\Sigma_q^{\dim-R}$ with $\run{\bfy} = r$. This concludes the proof for arbitrary $R$. 
		
		For $R=1$, we can use the refinement $\VD^q(\bfx,1) = \run{\bfx}$ and obtain, using the same arguments as above,
		\begin{align*}
		&  |\C| \geq q \sum_{r=1}^{\dim-1} \frac{(q-1)^{r-1}\binom{\dim-2}{r-1}}{r+2} & \\
		 & = q\sum_{r=1}^{\dim-1} \frac{(q-1)^{r-1}(\dim-2)!}{(r-1)!(\dim-r-1)!(r+2)} & \\
		 & = \frac{q}{(\dim-1)\dim(\dim+1)} \sum_{r=1}^{\dim-1} (q-1)^{r-1}r(r+1) \binom{\dim+1}{r+2} & \\
		 &\overset{(c)}{=}\hspace{-0.25ex}\frac{q}{(q\hspace{-0.25ex}-\hspace{-0.25ex}1)^2(\dim\hspace{-0.25ex}-\hspace{-0.25ex}1)\dim(\dim\hspace{-0.25ex}+\hspace{-0.25ex}1)} \sum_{r=3}^{\dim+1} (q\hspace{-0.25ex}-\hspace{-0.25ex}1)^{r\hspace{-0.25ex}-\hspace{-0.25ex}1}(r^2\hspace{-0.25ex}-\hspace{-0.25ex}3r\hspace{-0.25ex}+\hspace{-0.25ex}2)\binom{\dim\hspace{-0.25ex}+\hspace{-0.25ex}1}{r}, &
		\end{align*}
		where $(c)$ follows from a shift of the variable $r$ inside the sum. Using the equalities $$\sum_{i=0}^{m}\binom{m}{i}x^{i-1} = (x+1)^{m}/x, \sum_{i=0}^{m}i \binom{m}{i}x^{i-1} = m(x+1)^{m-1},$$ and $$\sum_{i=0}^{m}i(i-1)\binom{m}{i} x^{i-1}= m(m-1)x(1+x)^{m-1},$$ we obtain by standard transformations after a few steps
		$$\KD^q(n,1) \geq \frac{q^{\dim}(n-2)}{(q-1)n(n+1)}.$$
		Finally, it remains to derive the asymptotic bound for fixed $R$ and large $\dim$. 
		We will use the following inequality, 
		\[q\sum_{r=1}^{\dim-R} \frac{(q-1)^{r-1}\binom{\dim-R-1}{r-1}}{\binom{r+3R-1}{R}} \geq \frac{R!q^{n} }{n^R(q-1)^{R}}\left( 1 -o(1)  \right),\]	
whose standard, but rather technical, proof is deferred to Proposition~\ref{prop:deletion:approx:lb} in the appendix.	
	\end{proof}

\section{Single-Insertion/Deletion-Covering Codes}

Having lower bounds on the sizes of $R$-insertion- and $R$-deletion-covering codes in hand, we now prove existence of these codes for single insertions and deletions for both binary and non-binary alphabets. Surprisingly, finding covering codes for insertions is harder than for deletions, which is in sharp contrast with error-correcting codes, which have been proven to be equivalent for insertion and deletion errors \cite{Lev65}. For the case of deletions, we prove the existence of codes using explicit constructions. For the case of insertions, due to the lack of small explicit constructions, we resort to proving the existence of codes based on a random construction and a recursive construction.

\subsection{Single-Deletion-Covering Codes} \label{sec:single_del_nb}

Let us recall the definition of the well-known Varshamov-Tenengolts (VT)~\cite{VT65} codes and their role as single-deletion-covering codes. 
\begin{defn}
	For all $\dim$ and $0\leq a \leq \dim$, let $\C_\mathsf{VT}(n;a) \subseteq \bits$ be the {\bf Varshamov-Tenengolts code}
	$$ \C_\mathsf{VT}(n;a) = \left\{\bfc \in \bits: \sum_{i=1}^{\dim} ic_i \equiv a \bmod (\dim+1) \right\}. $$
\end{defn}
It is well known~\cite{Levenshtein1992} that for all $n$ and $a$ the Varshamov-Tenengolts code  $\C_\mathsf{VT}(n;a)$ is a perfect code under deletions, that is, $\bigcup_{\bfc\in\C_\mathsf{VT}(n;a)}\BD(\bfc, 1) = \bit^{\dim-1}$. The next corollary is a direct result of this important property.
\begin{corollary} \label{cor:single:deletion:covering}
	For all $n\geq 1$ it holds that 
	$$\KD(n,1)\leq  \frac{2^{\dim}}{\dim+1}. $$
\end{corollary}
It is also known that the largest (resp. smallest) of the VT codes is achieved for $a=0$ (resp. $a=1$)~\cite{G67}. Hence, while for deletion-correcting codes it is common to choose the code $\C_\mathsf{VT}(n;0)$, for the purpose of minimizing the size of single-deletion-covering codes, one should choose the code $\C_\mathsf{VT}(n;1)$. 

Unfortunately, the same property does not hold for insertions, i.e., the VT code is not a perfect code for insertions. In fact, this can be verified by simple counting arguments using the VT code size and the single-insertion ball size: as shown in Theorem~\ref{thm:lower:bound:insertion}, a lower bound on the size of any single-insertion-covering code is $2^{n+1}/(n+2)$, which is roughly twice the size of the VT codes. It can further be seen that,  while  the tasks of correcting a fixed number of insertions, deletions, or a combination of insertions and deletions are all equivalent~\cite{L66}, this sort of equivalence does not extend to covering codes. This makes the problem of finding good single-insertion-covering codes an intriguing question that will be addressed in Section~\ref{sec:single}.


VT codes have  a non-binary extension, presented by Tenengolts in~\cite{T84}, which can  correct a single deletion in the non-binary case. However, this family of codes is no longer perfect. In fact, their guaranteed size is $\frac{q^n}{qn}$, while the upper bound on a single-deletion-correcting code is approximately $\frac{q^n}{(q-1)n}$. This is also roughly the lower bound on a non-binary single-deletion-covering code we derived in Theorem~\ref{thm:lower:bound:deletions}, which confirms that these codes are not perfect and, therefore, are not single-deletion-covering codes. Our main result in this section is another non-binary extension of the binary VT codes, which we will show  does satisfy the single-deletion covering property.

For an integer $m$, we denote by $(m)_2$ the value of $(m\bmod 2)$ and for a vector $\bfx=(x_1,\ldots,x_n)$, let $(\bfx)_2 = ((x_1)_2,\ldots,(x_n)_2)$.
\begin{defn}\label{def:non_binary_VT}
	For all positive $\dim$, $q\geq 2$, $0\leq a \leq \dim$, and $ 0\leq b< \lfloor q/2\rfloor$,  let $\C_\mathsf{NBVT}^q(n;a,b) \subseteq \Sigma_q^n$ be the code 
	\begin{align*}
	\C_\mathsf{NBVT}^q(n;a,b)  = \bigg\{\bfc \in \Sigma_q^n \ : \ & \sum_{i=1}^{\dim} i(c_i)_2 \equiv a \bmod (\dim+1), & \\
	& \sum_{i=1}^{\dim} \left\lfloor \frac{c_i}{2} \right\rfloor \equiv b \bmod \left(\left\lfloor \frac{q}{2}\right\rfloor\right) \bigg\}. & 
	\end{align*}
\end{defn}
The following theorem proves that the code $\C_\mathsf{NBVT}^q(n;a,b)$ is indeed a non-binary single-deletion-covering code.
\begin{thm} \label{thm:single:q:deletion}
	For all positive $\dim$, $q\geq 2$, $0\leq a \leq \dim$, and $ 0\leq b< \lfloor q/2\rfloor$,  the code $\C_\mathsf{NBVT}^q(n;a,b)$ is a single-deletion-covering code.  Furthermore,
	$$\KD^q(n,1)\leq \frac{q^{\dim}}{(\dim+1)\lfloor q/2\rfloor}. $$
\end{thm}
\begin{proof}
	Let $\bfx =(x_1,\ldots,x_{n-1})\in\Sigma_q^{\dim-1}$. Since the binary VT code $\C_\mathsf{VT}(n;a)$ is a covering code, it follows that there exist $1\leq i\leq n$ and a binary value $d$ such that 
	$$(x_1,\ldots,x_{i-1},d, x_{i},\ldots,x_{n-1})_2\in \C_\mathsf{VT}(n;a).$$ Let 
	$$s= \left( b- \sum_{i=1}^{\dim-1} \left\lfloor \frac{x_i}{2} \right\rfloor\right) \bmod \left(\left\lfloor \frac{q}{2}\right\rfloor\right).$$
	Then, it holds that 
	$$\bfc=(x_1,\ldots,x_{i-1},2s+d, x_{i},\ldots,x_{n-1}) \in \C_\mathsf{NBVT}^q(n;a,b),$$
	and $\bfx\in  \BD(\bfc, 1)$. 
\end{proof}
Lastly, we note that this construction improves upon the construction in~\cite{afrati-anchor}\footnote{The result is stated in Corollary 5.5 in~\cite{afrati-anchor}. However, note that the authors of this paper refer to deletion-covering codes as insertion-covering codes and the result is stated over length-$(n+1)$ codes.}, which provides single-deletion-covering codes of size  $q^n/n$.

\subsection{Single-Insertion-Covering Codes} \label{sec:single}

In this section, we study single-insertion-covering codes. Our main result is stated in the following theorem.

\begin{thm}\label{thm:one:insertion}
For all $\dim \geq 1$ it holds that
$$\KI^q(n,1) \leq \mu_{\mathsf{I}} \frac{q^{n+1}}{(n+1)(q-1)+1},$$ 
where $\mu_{\mathsf{I}} \leq 7$.
\end{thm}

	Note that our result is stated as a fraction of the sphere-covering lower bound in Theorem \ref{thm:lower:bound:insertion} and implies that the size of optimal single-insertion covering codes is at most a factor of $7$ from the theoretical lower limit. Our proof is inspired by and follows the strategy of the existential construction of asymmetric covering codes due to Cooper, Ellis, and Kahng~\cite{cooper2002asymmetric}. The argument proceeds in two main steps. First, we use a random subset $S \subseteq \Sigma_q^{n_1}$ of an appropriate size to cover all but a small fraction of words $T \subseteq \Sigma_q^{n_1+1}$ with a single insertion. (This is analogous to the patched covering code in \cite{cooper2002asymmetric}.) Then, we ``fix up'' the set $S$ using a ``good'' single-insertion-covering code to generate a covering code of larger codeword length. By picking the size of $S$ and $T$ appropriately and using good codes inductively, we show that we will not have to pay too much in efficiency in this process. 
	
	We begin by introducing the set operation that will be used in the ``fixing up'' operation. The main utility of this tensorization is that it allows us to handle the uncovered words in an efficient manner.
	\begin{lemma}\label{lem:semidir}
		Let $S\subseteq \Sigma_q^{n_1},T \subseteq \Sigma_q^{n_1+1}$ be such that $S$ covers $\Sigma_q^{n_1+1} \setminus T$ with a single insertion. Let $\C_{n_2} \subseteq \Sigma_q^{n_2}$ be a single-insertion-covering code. Then, the code
		\[ (S \otimes \Sigma_q^{n_2+1}) \cup (T \otimes \C_{n_2})  \]
		is a single-insertion-covering code of length \mbox{$n_1+n_2+1$} and of size at most $|S|\cdot q^{n_2+1} + |T|\cdot|\C_{n_2}|$, where  \mbox{$A\otimes B=\{\bfa\bfb:\bfa\in A, \bfb \in B\}$} is the tensor product of two sets and $\bfa\bfb$ is the concatenation of $\bfa$ and $\bfb$
	\end{lemma}
	\begin{proof}
		Consider any $\bfx\bfy \in \Sigma_q^{n_1+n_2+2}$, where $\len{\bfx} = n_1+1$ and $\len{\bfy}=n_2+1$. We consider two cases. If $\bfx$ is covered by $\bfs \in S$, then $\bfx\bfy$ is covered by $\bfs\bfy \in S \otimes \Sigma_q^{n_2+1}$. Otherwise, $\bfx \in T$. In this case, let $\bfc \in \C_{n_2}$ be the word covering $\bfy \in \Sigma_q^{n_2+1}$. Then, $\bfx\bfy$ is covered by $\bfx\bfc$, and $\bfx\bfc \in T \otimes\C_{n_2}$. The size of the code directly follows from the definition of the tensorization and the union bound.
	\end{proof}
	We next find a suitable $(S,T)$ pair by randomly selecting the subset $S$. The words in $S$ are non-uniformly sampled from $\Sigma_q^{n_1}$, which  will reduce the overall code size by a constant factor compared to uniform sampling. The intuitive motivation for this is that  some words in  $\Sigma_q^{n_1+1}$ are harder to cover because their single-deletion balls are smaller. Non-uniform sampling ensures that  the words in $S$  cover words in  
	$\Sigma_q^{n_1+1}$ in a more equitable fashion.  
	
	The following lemma provides a bound on the sizes of $S$ and $T$. Although we could bound the sizes of $S$ and $T$ directly, the formulation in the lemma scales the size of  the uncovered set $T$ by $\mu_{\mathsf{I}}/\VI^q(n_2,1)$ because this is the factor saved by the use of induction later in the construction.
	
	\begin{lemma} \label{lem:existence:almost:covering}
		For all $n \geq 1$ there exist integers $n_1,n_2$ with $n_1+n_2+1=n$ and sets $S \subseteq \Sigma_q^{n_1},T \subseteq \Sigma_q^{n_1+1} $ such that $S$ covers $\Sigma_q^{n_1+1} \setminus T$, that is, $T = \Sigma_q^{n_1+1} \setminus \bigcup_{\bfs \in S} \BI^q(\bfs,1)$, while the sizes of $S$ and $T$ satisfy
		$$ |S| + \frac{\mu_{\mathsf{I}} |T|}{\VI^q(n_2,1)} \leq \frac{\mu_{\mathsf{I}} q^{n_1+1}}{\VI^q(n,1)}, $$
		where $\mu_{\mathsf{I}} \leq7$.
	\end{lemma}
	\begin{proof}
		For $n \leq \frac{q \mu_{\mathsf{I}}-q}{q-1}$, the statement is fulfilled by $S = \Sigma_q^{n_1}$ and $T = \emptyset$. Assume that $n > \frac{q \mu_{\mathsf{I}}-q}{q-1}$.  We prove the existence of an $(S,T)$ pair with sizes satisfying the lemma by means of a random construction. Include each word $\bfx \in \Sigma_q^{n_1}$ in $S$ with probability $q_\bfx \overset{\mathrm{def}}{=} c\VD^q(\bfx,1)^{-1}$ for a constant $c>0$ to be set later. Let $T$ be all remaining words that are not covered by $S$, i.e., $T = \Sigma_q^{n_1+1} \setminus \bigcup_{\bfs \in S} \BI^q(\bfs,1)$. 
		
		For a fixed word $\bfy \in \Sigma_q^{n_1+1}$, we have that $\bfy$ is covered by $S$ unless all of the words covering $\bfy$ fail to be included in $S$. The number of words that can cover $\bfy$ is exactly $\VD^q(\bfy,1)$, the size of the single-deletion ball. Note that $\VD^q(\bfy,1) = \run{\bfy}$ \cite{Lev65}, and observe that for any $\bfx \in \BD^q(\bfy,1)$ the number of runs cannot increase as a result of the deletion, i.e., $\run{\bfx}\leq \run{\bfy}$. Hence, $q_\bfx = c\VD^q(\bfx,1)^{-1} = c\run{\bfx}^{-1} \geq c \run{\bfy}^{-1} \overset{\mathrm{def}}{=} q_\bfy$. 
		We bound the probability that $S$ misses $\bfy$ as follows:
\begin{align*}
\mathrm{P}[\bfy \text{ is uncovered}] & = \prod_{\bfx \in \BD^q(\bfy,1)}(1-q_\bfx) & \\
& \overset{(a)}{\leq} \prod_{\bfx \in \BD^q(\bfy,1)}(1-q_\bfy)  = (1-q_\bfy)^{\VD^q(\bfy,1)}, 
\end{align*}
		where $(a)$ uses that $q_\bfx \geq q_\bfy$, as discussed above. 
		
		We now compute the expected weighted size of $S$ and $T$ under the above random selection.
\begin{align*}
 W & \overset{\mathrm{def}}{=} \mathrm{E}\left[|S| + \frac{\mu_{\mathsf{I}} |T|}{\VI^q(n_2,1)}\right] = \mathrm{E}[|S|] + \frac{\mu_{\mathsf{I}}\mathrm{E}[|T|]}{\VI^q(n_2,1)}  & \\
 & = \sum_{\bfx \in \Sigma_q^{n_2}} q_\bfx + \frac{\mu_{\mathsf{I}}}{\VI^q(n_2,1)}\sum_{\bfy \in \Sigma_q^{n_2+1}} \mathrm{P}[\bfy \text{ is uncovered}].  & 
 \end{align*}
		Plugging in the bound for $\mathrm{P}[\bfy \text{ is uncovered}]$ and recalling that $q_\bfx = c\run{\bfx}^{-1}$, we obtain
		$$ W \leq \sum_{\bfx \in \Sigma_q^{n_1}} \frac{c}{\run{\bfx}} + \frac{\mu_{\mathsf{I}}}{\VI^q(n_2,1)}\sum_{\bfy \in \Sigma_q^{n_1+1}} (1-q_\bfy)^{\VD^q(\bfy,1)}. $$
		It is well-known~\cite{Lev65} that the number of words $\bfx \in \Sigma_q^{n_1}$ with $\rho(\bfx) = r$ is given by $q\binom{n_1-1}{r-1}(q-1)^{r-1}$, which allows us to group terms in the first sum by $\rho(\bfx)=r$. Using that $1-z \leq \mathrm{e}^{-z}$ for all $z \in \R$, we have that $(1-q_\bfy)^{\VD^q(\bfy,1)} \leq \mathrm{e}^{-c}$ and we bound $W$ by
		\begin{align*}
		W &\leq qc\sum_{r=1}^{n_1} \frac{\binom{n_1-1}{r-1}(q-1)^{r-1}}{r} + \frac{\mu_{\mathsf{I}}}{\VI^q(n_2,1)}\sum_{\bfy \in \Sigma_q^{n_1+1}} \mathrm{e}^{-c}
		 \\ &\overset{(b)}{=} qc\sum_{r=1}^{n_1} \frac{(n_1-1)!(q-1)^{r-1}}{r!(n_1-r)!} + \frac{q^{n_1+1}\mu_{\mathsf{I}} \mathrm{e}^{-c}}{\VI^q(n_2,1)} 
		\\ &\overset{(c)}{=} \frac{qc}{n_1(q-1)}\sum_{r=1}^{n_1} \binom{n_1}{r}(q-1)^r + \frac{q^{n_1+1}\mu_{\mathsf{I}} \mathrm{e}^{-c}}{\VI^q(n_2,1)},
		\end{align*}
		where in equality $(b)$ and $(c)$ we used the definition of the binomial coefficient $\binom{n}{k} = \frac{n!}{k!(n-k)!}$. Finally, we use the binomial identity $\sum_{k=0}^{n} \binom{n}{k}x^k = (1+x)^n $ and obtain
		\begin{align*}
		 W& \leq \frac{cq^{n_1+1}}{n_1(q-1)} + \frac{q^{n_1+1}\mu_{\mathsf{I}} \mathrm{e}^{-c}}{\VI^q(n_2,1)} \\
		& = \frac{\mu_{\mathsf{I}} q^{n_1+1}}{\VI^q(n,1)} \left( \frac{c\VI^q(n,1)}{\mu_{\mathsf{I}}n_1(q-1)} + \frac{\VI^q(n,1){e}^{-c}}{\VI^q(n_2,1)}\right).
		 \end{align*}
		Abbreviating the term in round brackets by $\gamma$ and setting $n_1 = \lfloor\beta n\rfloor$ for some $0\leq\beta\leq1$, we derive the upper bound
		\begin{align*}
			\gamma &\overset{\mathrm{def}}{=}  \frac{c\VI^q(n,1)}{\mu_{\mathsf{I}}n_1(q-1)} + \frac{\VI^q(n,1){e}^{-c}}{\VI^q(n_2,1)} \\
			&\leq \frac{\VI^q(n,1)}{n(q-1)} \left(\frac{cn}{\mu_{\mathsf{I}}n_1} + \frac{{e}^{-c}n}{n-n_1}\right)
			\\&\overset{(d)}{\leq}  \frac{\mu_{\mathsf{I}}}{\mu_{\mathsf{I}}-1} \left(\frac{cn}{\mu_{\mathsf{I}}\lfloor\beta n\rfloor} + \frac{{e}^{-c}n}{n-\lfloor\beta n\rfloor}\right),
		\end{align*}
		where we used in equality $(d)$ that $\VI^q(n,1)/(n(q-1))$ is monotonically decreasing in $n$ and thus $\VI^q(n,1)/(n(q-1)) \leq \mu_{\mathsf{I}}/(\mu_{\mathsf{I}}-1)$ for all $n > (q \mu_{\mathsf{I}}-q)/(q-1)$. Note that this bound is convenient to handle as it is independent of $q$. To conclude, we find the smallest $\mu_{\mathsf{I}}$ such that there exists some $c>0$ and $0\leq \beta\leq1$ for which $ \gamma\leq 1$ for all $n > (q \mu_{\mathsf{I}}-q)/(q-1)$. A quick computer search yields that $\mu_{\mathsf{I}}=7$, $c=3$ and $\beta = \frac34$ fulfills this requirement. By definition of the random sets $S$ and $T$, any realization of them will have the desired property that $S$ covers $\Sigma_q^{n_1+1} \setminus T$. As the expected weighted size $W$ is at most $\frac{\mu_{\mathsf{I}} q^{n_1+1}}{\VI^q(n,1)}$, it  follows that there exists an $(S,T)$ pair satisfying the desired bound.
	\end{proof}

	Putting everything together, we prove Theorem~\ref{thm:one:insertion} for single-insertion-covering codes.
	\begin{proof}[Proof of Theorem~\ref{thm:one:insertion}]
		We proceed by induction on $\dim$. As the base case, for all $\dim \leq \frac{q\mu_{\mathsf{I}}-q}{q-1}$, it suffices to take $\C_\dim = \Sigma_q^n$. Assume now that the statement is correct for all lengths up to $\dim-1$, so that there exist codes $\C_{n_2}$ with size at most $\mu_{\mathsf{I}}q^{n_2+1}/\VI^q(n_2,1)$ for all $1 \leq n_2 \leq n-1$. Let $n_1+n_2+1=n$ and $S \subseteq \Sigma_q^{n_1}$ and $T \subseteq \Sigma_q^{n_1+1}$ denote sets guaranteed by Lemma~\ref{lem:existence:almost:covering}. Note that clearly $n_2<n$ in Lemma \ref{lem:existence:almost:covering}, which will be useful later. As these sets $S$ and $T$ fulfill the requirement of Lemma \ref{lem:semidir}, we define
		$$\C_{\dim} = (S \otimes \Sigma_q^{n_2+1}) \cup (T \otimes \C_{n_2}),$$
		and we have that
		there exists a single-insertion-covering code $\C_{\dim} \subseteq\Sigma_q^{\dim}$ of size 
		\begin{align*}
	 	|C_\dim| & \leq q^{n_2+1}|S| + |T|\cdot|\C_{n_2}| 
		\overset{(e)}{\leq} q^{n_2+1}\left(|S| + \frac{\mu_{\mathsf{I}}}{\VI^q(n_2,1)}|T|\right),
		\end{align*}
		where in $(e)$ we used the existence of a covering code of length $n_2<n$ and size $\mu_{\mathsf{I}}q^{n_2+1}/\VI^q(n_2,1)$ by the induction hypothesis. Using the existence of good sets $S$ and $T$ from Lemma \ref{lem:existence:almost:covering}, we obtain the desired bound on the code size 
		\[|\C_{\dim}| 
		\leq q^{n_2+1}\frac{\mu_{\mathsf{I}} q^{n_1+1}}{\VI^q(n,1)} = \frac{\mu_{\mathsf{I}} q^{n+1}}{\VI^q(n,1)}.\]
	\end{proof}
Together with our existence result from Theorem \ref{thm:one:insertion}, we can infer that the size of the smallest single-insertion-covering code lies between $q^{n+1}/\VI^q(n,1)$ and $7q^{n+1}/\VI^q(n,1)$ and thus is known up to a constant factor of $7$.

\section{Multiple-Insertion/Deletion-Covering Codes}\label{sec:multiple}

We now turn to the discussion of multiple-insertion/deletion covering codes. We begin by defining the optimal density of insertion- and deletion-covering codes, by analogy with the
notion of density often used in the context of classical covering codes. 
\begin{defn}
	  For $R$-insertion-covering codes of length $n$,  the {\bf optimal density $\mu_{\mathsf{I}}^q(n,R) $} is defined as
	\begin{align*}
	\mu_{\mathsf{I}}^q(n,R) &= \frac{\KI^q(n,R)\VI^q(n,R)}{q^{n+R}}. 	 
	\end{align*}
For  $R$-deletion-covering codes of length $n$, we define the {\bf optimal density $\mu_{\mathsf{D}}^q(n,R) $} as 
\begin{align*}
	\mu_{\mathsf{D}}^q(n,R) &= \frac{\KD^q(n,R)n^R(q-1)^R}{q^nR!}.
	\end{align*}
	Finally, for fixed $R$, we define the corresponding {\bf asymptotic optimal densities} 
$ \mu_{\mathsf{I}}^{q,*}(R)$  and $\mu_{\mathsf{D}}^{q,*}(R)$ as  
	\begin{align*}
	\mu_{\mathsf{I}}^{q,*}(R) &= \limsup_{n\rightarrow \infty} \mu_\mathsf{I}^q(n,R)
\end{align*}
and 
\begin{align*}
	\mu_{\mathsf{D}}^{q,*}(R) &= \limsup_{n\rightarrow \infty} \mu_\mathsf{D}^q(n,R).
	\end{align*}
\end{defn}
Note that   we define the optimal density of deletion-covering codes slightly differently than that of 
insertion-covering codes. This is due to the fact that for deletions, the deletion balls are non-uniform and the density is thus defined with respect to the lower bound obtained in Theorem~\ref{thm:lower:bound:deletions} for large $n$. A powerful tool in building covering codes of larger radius  is to take the tensor product of two short covering codes of small radius. For example, taking the tensor product of two covering codes of length $n$ and radius $1$ gives a covering code of length $2n$ and radius $2$. However, a straightforward application of this technique only gives covering codes whose density is at least exponential in $R$. We therefore refine this technique to obtain codes that have a density that is almost linear in $R$. Note that in the following sections we prove our results for binary words for simplicity. The proofs for $q>2$ are obtained by only a slight modification, which will be explained in more detail in Remark \ref{rem:extension:non-binary} at the end of Section \ref{subsec:multiple:deletion:covering:codes}.
\subsection{Multiple-Insertion-Covering Codes} 
Our main result about $R$-insertion-covering codes is stated in the following theorem.

\begin{thm} \label{thm:k:insertion}
	For any fixed $R\geq 2$ and $q \geq 2$,
	$$ \mu^{q,*}_{\mathsf{I}}(R) \leq \mathrm{e}(R \log R + \sqrt{2R\log R} +1)\mu^{q,*}_\mathsf{I}(1). $$
\end{thm}
Recall that according to Theorem~\ref{thm:one:insertion}, we have that $\mu^{q,*}_\mathsf{I}(1) \leq 7$. Before proving the theorem, we give a short outline of the proof, along with the intuition behind it. As in the proof of the upper bound for single-insertion-covering codes, we start  by proving in Lemma~\ref{lem:existence:almost:covering:R} the existence of a small \emph{almost-covering} code $S$, i.e., a code that covers all words in $\{0,1\}^{n+R}$ except for a small subset  $T$. Then, in Lemmas~\ref{lem:semidir:k} and \ref{lemma:recursive:R:insertion}, we combine this code with small covering codes to  recursively build larger codes.  By computing the size of the resulting codes, we can then prove Theorem~\ref{thm:k:insertion}. 

The proof of  the existence of small almost-covering codes is again based on a random coding argument. Since  we are building covering codes for insertions, we must take into account the fact that each word $\bfy \in \{0,1\}^{n+R}$ is covered by a different number of potential codewords $\bfx \in \{0,1\}^n$. This is because the number of words that can cover $\bfy$ is given by $\VD(\bfy,R)$, which is known to depend on $\bfy$. In our random selection of codewords, we  therefore need to favor codewords that cover words with small $\VD(\bfy,R)$ to ensure that each  word is covered with high enough probability. Our proof follows the general idea of a recursive covering code construction presented in~\cite{krivelevich2003covering}, here modified to work for insertions. In particular, we need to adapt the arguments for the random construction and the recursive combination of almost-covering codes with existing covering codes. 

The following lemma  gives an upper bound on the sizes of the almost-covering code  $S$ and the complement $T$ of its coverage. 
\begin{lemma}\label{lem:existence:almost:covering:R}
	For every $n\geq R$ and every positive constant $c>0$ there exists a set $S\subseteq \bit^{\dim-R}$ of size at most $$|S| \leq \frac{c2^{\dim}}{\VI(n-R,R)} f_{n,R}$$ such that $S$ covers $\bit^{\dim} \setminus T$ with $R$ insertions for some set $T \subseteq \bit^{\dim}$ of size at most
	$$|T| \leq \mathrm{e}^{-c}2^{\dim},$$
	for some function $f_{n,R}$ with $\lim_{n\rightarrow \infty} f_{n,R} = 1$.
\end{lemma}
\begin{proof}
	We prove the lemma by choosing  a random set $S$ and computing the expected number of words that are not covered by such a random choice. Let $\X_i = \{ \bfx \in \{0,1\}^{\dim-R}: \VD(\bfx,R) = i \}$ be the set of all strings of length $n$, which have a deletion ball size of exactly $i$. We construct $S$ by choosing $S = S_1 \cup S_2 \cup\dots\cup S_m$, where $S_i \subseteq \X_i$ and $m \leq \binom{n}{R}$ is the maximum size of the deletion ball  of any $x \in \bit^{\dim-R}$. Denoting $m_i=|\X_i|$, each $S_i$ is a uniformly chosen random subset of 
$\X_i$ of cardinality $|S_i| = \lceil c m_i/i \rceil$, if $c m_i/i \leq m_i$, and $m_i$ otherwise. Hereby each such subset has the same probability. By this choice of the sets $S_1,\dots,S_m$, the probability that  any $\bfy \in \bits$ is not covered by $S$ can be bounded from above as follows. First, note  that
	$$ \mathrm{P}[\bfy \text{ is uncovered}] = \prod_{i=1}^{m}\mathrm{P}[S_i \cap \BD(\bfy,R) = \emptyset], $$
	since the random sets $S_i$ are independent. Denote by $\gamma_i = |\{ \X_i \cap \BD(\bfy,R) \}|$ the number of words in $\X_i$  which can cover $\bfy$.  With this notation,  the individual probabilities in the product can be expressed as 
	$$ \mathrm{P}[S_i \cap \BD(\bfy,R) = \emptyset] =  \frac{\binom{m_i-\gamma_i}{|S_i|}}{\binom{m_i}{|S_i|}} = \frac{ \prod_{j=0}^{|S_i|} (m_i-\gamma_i-j)}{\prod_{j=0}^{|S_i|} (m_i-j)}. $$
	From  the fact that $\frac{m_i-\gamma_i-j}{m_i-j} \leq \frac{m_i-\gamma_i}{m_i}$ for any $0\leq j <m_i-\gamma_i$, we obtain
	\begin{align*}
		\mathrm{P}[\bfy \text{ is uncovered}] &\leq \prod_{i=1}^{m} \left( \frac{m_i-\gamma_i}{m_i} \right)^{|S_i|} = \prod_{i=1}^{m} \left( 1-\frac{\gamma_i}{m_i} \right)^{|S_i|} \\
		&\overset{(a)}{\leq} \prod_{i=1}^{m} \mathrm{e}^{-\frac{\gamma_i}{m_i}|S_i|} \leq  \mathrm{e}^{-\sum_{i=1}^{m}c\frac{\gamma_i}{i}},
	\end{align*}
	where we used in $(a)$ that $1-x \leq \mathrm{e}^{-x}$ for any $x \in \mathbb{R}$. Let $\mu(\bfy) = \max_{\bfx \in \BD(\bfy,R)}\VD(\bfx,R)$ be the maximum deletion ball size of any $\bfx \in \BD(\bfy,R)$, which is obtained from $\bfy$ by $R$ deletions. Since $\gamma_i = 0$ for all $i> \mu(\bfy)$, we can bound the exponent from below by
	$$ \sum_{i=1}^{m}\frac{\gamma_i}{i} = \sum_{i=1}^{\mu(\bfy)}\frac{\gamma_i}{i} \geq \sum_{i=1}^{\mu(\bfy)} \frac{\gamma_i}{\mu(\bfy)} =  \frac{\sum_{i=1}^{\mu(\bfy)}\gamma_i}{\mu(\bfy)} = \frac{\VD(\bfy,R)}{\mu(\bfy)} \overset{(b)}{\geq} 1,  $$
	where $(b)$ follows from the fact that $\VD(\bfx,R) \leq \VD(\bfy,R)$ for any $\bfx \in \BD(\bfy,R)$. Hence, 
$\mathrm{P}[\bfy \text{ is uncovered}] \leq \mathrm{e}^{-c}$ and the expected size of $T$ is consequently at most $\mathrm{E}[|T|] \leq \mathrm{e}^{-c}2^\dim$. Thus, there must exist a set $S$ for which $|T|\leq \mathrm{e}^{-c}2^\dim$. 

It remains to compute the size of $S$. By construction
	$$ |S| = \sum_{i=1}^m |S_i| \leq cm+  c \sum_{i=1}^m \frac{m_i}{i} = cm + \sum_{\bfx \in \{0,1\}^{n-R}}\frac{c}{\VD(\bfx,R)}. $$
	Noting that  $\VD(\bfx,R) \geq \binom{\rho(\bfx)-R}{R}$ and recalling that the number of words of length $n$ with $r$ runs is   $2\binom{n-1}{r-1}$,   we obtain
	\begin{align*}
	|S| &\leq cm + c\sum_{r=1}^{n-R} \frac{2\binom{n-R-1}{r-1}}{\max\left\{ 1, \binom{r-R}{R} \right\}} \\
	&\leq c\binom{n}{R} + 2c \sum_{r=1}^{r^*} \binom{n-R-1}{r} +2c\sum_{r=r^*}^{n-R} \frac{\binom{n-R-1}{r-1}}{\binom{r^*-R}{R}} \\
	& \leq c\binom{n}{R} + 2c \sum_{r=1}^{r^*} \binom{n-R-1}{r} +c \frac{2^{n-R}}{\binom{r^*-R}{R}},
	\end{align*}
where  $r^*=\max\{2R,\frac{n}{2}-\sqrt{Rn\log n}\}$.
	For $r^* = 2R$, we can directly bound the second term by
	$$ 2c \sum_{r=1}^{r^*} \binom{n-R-1}{r} \leq 2c\binom{n-R+r^*}{r^*} = 2c\binom{n+R}{2R}. $$
	On the other hand, for $r^*=\frac{n}{2}-\sqrt{Rn\log n}$, applying Chernoff's inequality to the binomial tail, we obtain
	$$	\sum_{r=1}^{r^*} \binom{n-R-1}{r} \leq  \sum_{r=1}^{r^*} \binom{n}{r} \leq 2^{n} \mathrm{e}^{-2\frac{(n/2-r^*)^2}{n}} = \frac{2^n}{n^{2R}}. $$
	Hence, the overall size of $S$ is bounded from above by
	$$ |S| \leq c\frac{2^n}{\VI(n-R,R)} f_{n,R}, $$  
	where $f_{n,R}$ is at most
	\begin{align*}
	f_{n,R} \leq & ~\frac{\binom{n}{R}\binom{n+R}{R}}{2^n} + 2\binom{n+R}{R}\max\left\{ \frac{\binom{n+R}{2R}}{2^n},\frac{1}{n^{2R}}  \right\} \\
	& + \frac{2^{-R} \binom{n+R}{R}}{\binom{r^*-R}{R}},
	\end{align*}
	where we used that $\VI(n-R,R) \leq \binom{n+R}{R}$. 

Lastly, we bound the third summand in the bound on $|S|$. For $r^*=2R$ it is trivially bounded by $2^{-R}\binom{n+R}{R}$. When $r^*=\frac{n}{2}-\sqrt{Rn\log n}$, a quick calculation yields
	\begin{align*}
	\frac{2^{-R}\binom{n+R}{R}}{\binom{r^*-R}{R}} &\leq \frac{2^{-R}\binom{n+R}{R}R!}{(n/2-\sqrt{Rn\log n}-2R)^R}\\
	& \overset{(a)}{\leq} \frac{2^{-R}n^R \mathrm{e}^{R^2/n}}{(\frac{n}{2})^R (1-2R\sqrt{R \log n}/\sqrt{n}-4R^2/n)} \\
	& = \frac{n\mathrm{e}^{R^2/n}}{n-2R\sqrt{Rn \log n}-4R^2},
	\end{align*}
	where in inequality $(a)$ we used that $(1+x)^R \geq 1+Rx$ for any $x \geq -1$. Finally we obtain for $\frac{n}{2} -\sqrt{Rn\log n} \leq 2R$,
	$$f_{n,R}\leq \frac{(n+R)^{2R}}{2^n} + \frac{2(n+R)^{3R}}{2^n} + \frac{\binom{n+R}{R}}{2^R},$$
	and for $\frac{n}{2} -\sqrt{Rn\log n} > 2R$,
	$$ f_{n,R}\leq \frac{(n+R)^{2R}}{2^n} +\frac{2}{n^{R}} + \frac{n\mathrm{e}^{R^2/n}}{n-2R\sqrt{Rn \log n}-4R^2}. $$
	Here we additionally used $\binom{n}{R} \leq \binom{n+R}{R} \leq (n+R)^R$. Note that for large enough $n$ and any fixed $R$, $ \frac{n}{2} -\sqrt{Rn\log n} > 2R$ and it is directly verified that  $\lim_{n\rightarrow \infty} f_{n,R} = 1$.
\end{proof}
Note that while the expression of $f_{n,R}$ looks quite involved, we are interested in its asymptotic behavior and it will only be important in the following that it approaches $1$ for large $n$.
\begin{lemma}\label{lem:semidir:k}
	Let $S\subseteq \bit^{\dim_1-R_1},T \subseteq \bit^{\dim_1}$ be such that $S$ covers $\bit^{\dim_1} \setminus T$ with $R_1$ insertions. Denote by $\C_1 \subseteq \bit^{\dim_2+R_1}$ an $R_2$-insertion-covering code of length $\dim_2+R_1$ and by $\C_2 \subseteq \bit^{\dim_2}$ an $R$-insertion-covering code of length $n_2$. We have that
	\[ (S \otimes \C_1) \cup (T \otimes \C_2) \]
	is an $R=R_1+R_2$-insertion-covering code of length $\dim=\dim_1+\dim_2$ with size at most $|S|\cdot|\C_1| + |T|\cdot|\C_2|$.
\end{lemma}
\begin{proof}
	Consider any $\bfx \bfy \in \bit^{\dim+R}$, where $\len{\bfx} = \dim_1$ and $\len{\bfy}=\dim_2+R$. We distinguish between two cases. First consider the case where $\bfx$ is covered by $\bfs \in S$ with $R$ insertions. Denote by $\bfc_1 \in \C_1$ the word that covers $\bfy \in \bit^{\dim_2+R}$ with $R_2$ insertions. Note that such a word always exists, as $\C_1$ is an $R_2$-insertion-covering code. Then $\bfx\bfy$ is covered by $\bfs\bfc_1 \in S \otimes \C_1$ with a total of $R=R_1+R_2$ insertions. Otherwise, $\bfx \in T$. In this case, let $\bfc_2 \in \C_2$ be the string covering $\bfy \in \bit^{\dim_2+R}$ with $R$ insertions. Then, $\bfx\bfy$ is covered by $\bfx\bfc_2$, and $\bfx\bfc_2 \in T \otimes\C_2$. The size of the code directly follows from the union bound.
\end{proof}
\begin{lemma} \label{lemma:recursive:R:insertion}
For any $n \geq R$ and $c>0$,
\begin{align*}
\mu_{\mathsf{I}}(n,R) \leq  & ~ c\mathrm{e}\mu_\mathsf{I}(n/R+R-1,1) \frac{(1+2R/n)^R}{1-R^2/n} f_{\frac{R-1}{R}n,R-1} \\
&+ R^R \mathrm{e}^{-c} \mu_{\mathsf{I}}(n/R,R) (1+2R/n)^R.
\end{align*}
\end{lemma}
\begin{proof}
Let $S\subseteq \bit^{\dim_1-R_1},T \subseteq \bit^{\dim_1}$ with $n_1\geq R_1$ be such that $S$ covers $\bit^{\dim_1} \setminus T$ with $R_1$ insertions. Denote by $\C_1 \subseteq \bit^{\dim_2+R_1}$ an $R_2$-insertion-covering code of length $\dim_2+R_1$ and by $\C_2 \subseteq \bit^{\dim_2}$ an $R$-insertion-covering code of length $n_2$, where $n_1+n_2=n$, $n_1 = \frac{y-1}{y}n$, $n_2 = \frac{n}{y}$, and $R_1+R_2=R$. We compute the size of the tensorization $(S \otimes \mathcal{C}_1) \cup (T \cup \mathcal{C}_2)$, which, by Lemma \ref{lem:semidir:k}, is an $R$-insertion covering code of length $n$. To begin with, $|(S \otimes \mathcal{C}_1) \cup (T \otimes \mathcal{C}_2)| \leq |S|\cdot|\C_1| + |T|\cdot|\C_2|$. Using Lemma \ref{lem:existence:almost:covering:R} and optimal codes $\mathcal{C}_1$ and $\mathcal{C}_2$, we can bound the size of $S$ and $T$ to obtain
\begin{align*}
|(S \otimes \mathcal{C}_1) \cup (T \otimes \mathcal{C}_2)| \leq &~ \frac{c 2^{n+R} f_{n_1,R_1} \mu_\mathsf{I}(n_2+R_1,R_2)}{\VI(n_1-R_1,R_1) \VI(n_2+R_1,R_2)} \\
&+\frac{ \mathrm{e}^{-c}2^{n+R}\mu_\mathsf{I}(n_2,R)}{\VI(n_2,R)}.
\end{align*}
Since $(S \otimes \mathcal{C}_1) \cup (T \otimes \mathcal{C}_2)$ is a covering code of length $n$ and covering radius $R$, we obtain
\begin{align*}
&\mu_\mathrm{I}(n,R) \hspace{-.08cm}\leq\hspace{-.08cm}  \frac{c \VI(n,R) f_{n_1,R_1} \mu_\mathsf{I}(n_2\hspace{-.08cm}+\hspace{-.08cm}R_1,R_2)}{\VI(n_1\hspace{-.08cm}-\hspace{-.08cm}R_1,R_1) \VI(n_2\hspace{-.08cm}+\hspace{-.08cm}R_1,R_2)} \hspace{-.08cm}+\hspace{-.08cm} \frac{\VI(n,R)\mu_\mathsf{I}(n_2,R)}{\mathrm{e}^{c}\VI(n_2,R)} \\
& \overset{(a)}{\leq} \frac{c R_1!R_2! (n+2R)^R \mu_\mathsf{I}(n_2+R_1,R_2)}{R!(n_1-R_1)^{R_1} (n_2+R_1)^{R_2} }f_{n_1,R_1} \\[-2ex]
& \qquad\qquad\qquad\qquad\qquad\qquad\quad+ \mathrm{e}^{-c} \left(\frac{n+2R}{n_2}\right)^R\mu_\mathsf{I}(n_2,R) \\
& \leq \frac{c n^R \mu_\mathsf{I}(n_2+R_1,R_2) }{n_1^{R_1} n_2^{R_2} \binom{R}{R_1}} \frac{(1+2R/n)^R}{1-R_1^2/n_1}f_{n_1,R_1} \\[-1ex]
& \qquad\qquad\qquad\qquad\quad+ \mathrm{e}^{-c} \left(\frac{n}{n_2}\right)^R  (1+2R/n)^R \mu_\mathsf{I}(n_2,R),
\end{align*}
where in $(a)$ we used the well-known inequalities $\binom{n+R}{R}\leq \VI(n,R) \leq \binom{n+2R}{R}$ and $(n-R)^R/R! \leq\binom{n}{R} \leq n^R/R!~$. Inserting $R_1=R-1$, $R_2=1$, and $y=R$ yields the lemma. 
\end{proof}
With this recursive expression, we are ready to prove the theorem with the help of the following lemma.
\begin{lemma}[cf. \cite{krivelevich2003covering}] \label{lemma:asym:density}
	Let $(\mu_n), (\mu_n'), (a_n)$ and $(b_n)$, $n \in \mathbb{N}$ be sequences of positive numbers with
	$$ \limsup_{n \rightarrow \infty } \mu_n' \leq \mu', \quad \limsup_{n \rightarrow \infty } a_n \leq a \quad \limsup_{n \rightarrow \infty } b_n \leq b $$
	and
	$$ \mu_n \leq a_n \mu'_{n/R} + b_n \mu_{n/R}, $$
	where $R >1$. Then
	$$ \limsup_{n \rightarrow \infty } \mu_n \leq \frac{a\mu'}{1-b}.$$
\end{lemma}
One convenient property of this lemma is that it incorporates the recursive assembly of the covering codes, without having to perform a thorough analysis of the induction start. We are now in a position to prove Theorem~\ref{thm:k:insertion}.
\begin{proof}[Proof of Theorem~\ref{thm:k:insertion}]
	Using Lemma \ref{lemma:recursive:R:insertion} and \ref{lemma:asym:density}, we obtain
	$\mu^*_{\mathsf{I}}(R) \leq \frac{c \mathrm{e} }{1-R^R \mathrm{e}^{-c}}\mu^*_{\mathsf{I}}(1)$. Minimizing $\frac{c \mathrm{e} }{1-R^R \mathrm{e}^{-c}}$ over $c$,  we can directly verify that $\min_c\frac{c \mathrm{e} }{1-R^R \mathrm{e}^{-c}} = \mathrm{e}(c_0+1)$, where $c_0$ is the solution to $c+1=\mathrm{e}^cR^{-R}$. Using standard bounds on $c_0$, we obtain the theorem.
\end{proof}
\subsection{Multiple-Deletion-Covering Codes} \label{subsec:multiple:deletion:covering:codes}
Proving the existence of small covering codes for deletions follows basically the same steps as the proof for the case of insertions. However, there are some subtle differences, such as the different definition of density for deletion-covering codes. Our main result is as follows.
\begin{thm} \label{thm:k:deletion}
	For any fixed $R\geq 2$ and $q \geq 2$,
	$$ \mu^{q,*}_{\mathsf{D}}(R) \leq \mathrm{e}(R \log R + \sqrt{2R\log R} +1) \mu^{q,*}_{\mathsf{D}}(1). $$
	In particular, for $q=2$,
	$$ \mu^{*}_{\mathsf{D}}(R) \leq \mathrm{e}(R \log R + \sqrt{2R\log R} +1). $$
\end{thm}
Note that compared to the case of insertions, we could use in Theorem \ref{thm:k:deletion} that for the binary case $\mu_{\mathsf{D}}^*(1)=1$, which results in a tighter bound also for the case of $R>1$. Since the outline of the proof for the case of deletions is similar to that of insertions, we merely state the ingredients of the proof to establish the analogy to the case of insertions. Again, we prove the theorem for the binary case, however the extension to non-binary alphabets is straightforward.
\begin{lemma}\label{lem:existence:almost:covering:R:del}
	For every $n$ and $R$ and every positive constant $c>0$ there exist sets $S\subseteq \bit^{\dim+R}$ and $T \subseteq \bit^{\dim}$ with
	$$|S| \leq \frac{c{2^{\dim+R}}}{\binom{n}{R}}$$
	such that $S$ covers $\bit^{\dim} \setminus T$ with $R$ deletions for some set $T$ of size at most
	$$|T| \leq \mathrm{e}^{-c} \mathrm{e}^{\frac{\VI(n,R)}{2^{n+R}}} 2^{\dim}.$$
\end{lemma}
\begin{proof}
	We prove the lemma using a random set $S$ and then compute the expected number of words that are not covered by such a random choice. We choose $S$ to be a uniformly random set of cardinality $|S| =  \left\lfloor c{2^{n+R}}/{\VI(n,R)} \right\rfloor $, where each subset has the same probability. By this choice of $S$, the probability for any $\bfy \in \bits$ to be not covered by $S$ is given by
	\begin{align*}
	\mathrm{P}[\bfy \text{ is uncovered}] &= \frac{\binom{2^{n+R}-\VI(n,R)}{|S|}}{\binom{2^{n+R}}{|S|}}  \leq \left( \frac{2^{n+R}-\VI(n,R)}{2^{n+R}} \right)^{|S|} \\
	&= \left( 1-\frac{\VI(n,R)}{2^{n+R}} \right)^{|S|} \leq \mathrm{e}^{-c} \mathrm{e}^{\frac{\VI(n,R)}{2^{n+R}}}.
	\end{align*}
	The bound on the size of $S$ follows from $\VI(n,R) \geq \binom{n}{R}$.
\end{proof}
\begin{lemma}\label{lem:semidir:k:del}
	Let $S\subseteq \bit^{\dim_1+R_1},T \subseteq \bit^{\dim_1}$ be such that $S$ covers $\bit^{\dim_1} \setminus T$ with $R_1$ deletions. Denote by $\C_1 \subseteq \bit^{\dim_2-R_1}$ an $R_2$-deletion-covering code of length $\dim_2$ and by $\C_2 \subseteq \bit^{\dim_2}$ an $R$-insertion-covering code of length $n_2$. We have that
	\[ (S \otimes \C_1) \cup (T \otimes \C_2) \]
	is an $R=R_1+R_2$-deletion-covering code of length $\dim=\dim_1+\dim_2$ with size at most $|S|\cdot|\C_1| + |T|\cdot|\C_2|$.
\end{lemma}
We omit the proof here as it is proven in the same manner as Lemma \ref{lem:semidir:k}. Using this construction of codes, we can again prove the following asymptotic relationship.
\begin{lemma} \label{lemma:recursive:R:del}
	For any $n \geq R$ and $c>0$,
	\begin{align*}
	\mu_{\mathsf{D}}(n,R) \leq  & ~ c\mathrm{e}\mu_\mathsf{D}(n/R+R-1,1) \gamma_\mathsf{D}(n,R) \\
	&+ R^R \mathrm{e}^{-c} \mu_{\mathsf{D}}(n/R,R) \gamma_\mathsf{D}'(n,R),
	\end{align*}
	for some functions $\gamma_\mathsf{D}(n,R)$ and $ \gamma_\mathsf{D}'(n,R)$ with 
	$$\lim_{n\rightarrow \infty} \gamma_\mathsf{D}(n,R)=\lim_{n\rightarrow \infty} \gamma_\mathsf{D}'(n,R)=1.$$
\end{lemma}

\begin{proof}
	Let $S\subseteq \bit^{\dim_1+R_1},T \subseteq \bit^{\dim_1}$ be such that $S$ covers $\bit^{\dim_1} \setminus T$ with $R_1$ deletions. Denote by $\C_1 \subseteq \bit^{\dim_2-R_1}$ an $R_2$-deletion-covering code of length $\dim_2-R_1$ and by $\C_2 \subseteq \bit^{\dim_2}$ an $R$-deletion-covering code of length $n_2$, where $n_1+n_2=n$, $n_1 = \frac{y-1}{y}n$, $n_2 = \frac{n}{y}$, and $R_1+R_2=R$. We compute the size of the tensorization $(S \otimes \mathcal{C}_1) \cup (T \otimes \mathcal{C}_2)$, which, by Lemma~\ref{lem:semidir:k:del},  is an $R$-deletion-covering code of length $n$. To begin with, $|(S \otimes \mathcal{C}_1) \cup (T \otimes \mathcal{C}_2)| \leq |S|\cdot|\C_1| + |T|\cdot|\C_2|$. Using Lemma~\ref{lem:existence:almost:covering:R:del}, we can bound the size of $S$ and $T$ to obtain
	\begin{align*}
	\KD(n,R) \leq&~  \frac{c 2^{n} \mu_\mathsf{D}(n_2-R_1,R_2)R_2!}{\binom{n_1}{R_1} (n_2-R_1)^{R_2}} \\
	&+\frac{ \mathrm{e}^{-c}\mathrm{e}^{\frac{\VI(n,R)}{2^{n+R}}}2^{n}\mu_\mathsf{D}(n_2,R)R!}{n_2^R}.
	\end{align*}
	Since $(S \otimes \mathcal{C}_1) \cup (T \otimes \mathcal{C}_2)$ is a covering code of length $n$ and covering radius $R$, we obtain
	\begin{align*}
	&\mu_\mathrm{D}(n,R)\hspace{-.08cm}\leq  \hspace{-.08cm}\frac{c \mu_\mathsf{D}(n_2\hspace{-.08cm}-\hspace{-.08cm}R_1,R_2)R_2!n^R}{\binom{n_1}{R_1} (n_2-R_1)^{R_2}R!} \hspace{-.08cm}+\hspace{-.08cm}\frac{\mathrm{e}^{-c} \mathrm{e}^{\frac{\VI(n,R)}{2^{n+R}}}\mu_\mathsf{D}(n_2,R)n^R}{n_2^R} \\
	& \leq \frac{c n^R \mu_\mathsf{D}(n_2-R_1,R_2) }{n_1^{R_1} n_2^{R_2} \binom{R}{R_1}} \frac{1}{(1-R_1^2/n_1)(1-R_1R_2/n_2)} \\[-1ex]
	& \qquad\qquad\qquad\qquad\qquad+ \mathrm{e}^{-c}  \left(\frac{n}{n_2}\right)^R \mu_\mathsf{D}(n_2,R)\mathrm{e}^{\frac{\VI(n,R)}{2^{n+R}}}.
	\end{align*}
	Inserting $R_1=R-1$, $R_2=1$, and $y=R$ yields the lemma. 
\end{proof}
\begin{remark} \label{rem:extension:non-binary}
	Proving Theorem \ref{thm:k:insertion} for non-binary words follows basically the same steps as for the binary case with only slight differences. To start with, the $q$-ary version of Lemma \ref{lem:existence:almost:covering:R} is obtained by replacing the binary expressions with their $q$-ary analogues, and letting $f_{n,R}$ depend also on $q$, but still converge to $1$ for $n \rightarrow \infty$. Lemma \ref{lem:semidir:k} directly extends to the non-binary case and Lemma \ref{lemma:recursive:R:insertion} can be extended by allowing additional factors in the recursive expression that depend on $q$ and converge to $1$ and using standard bounds on $\VI^q(n,R)$. The remaining steps are equivalent to the binary case. The proof of Theorem \ref{thm:k:deletion} for $q$-ary words is obtained analogously.
\end{remark}
\begin{remark}
	We note that Theorems \ref{thm:k:insertion} and \ref{thm:k:deletion} imply the asymptotic bounds in Table \ref{table:code-results}. More precisely, by the properties of the $\limsup$, for any $\epsilon >0$, there exists a value $n_0$, such that for all $n > n_0$, we have that $\mu^q_\mathsf{I}(n,R) \leq \mu^{q,*}_\mathsf{I}(R) (1+\epsilon)$  and $\mu^q_\mathsf{D}(n,R) \leq \mu^{q,*}_\mathsf{D}(R) (1+\epsilon)$.
\end{remark}

\section{Conclusion}\label{sec:conc}

This paper studied covering codes for insertions or deletions. We proved general sphere-covering lower bounds on the size of insertion- and deletion-covering codes.  We gave constructions for single-deletion- and single-insertion-covering codes that implied improved upper bounds on the code size. Finally, we presented upper bounds on the optimal density of multiple-insertion- and multiple-deletion-covering codes. 
	
There are many avenues for future work. There are gaps   between our lower  and upper  bounds for covering codes. For large covering radius $R$, we expect that much smaller codes should be possible, e.g., for $R = \epsilon \dim$ with $\epsilon \in (0,1/2)$.  Many of our proofs are existential in nature;  it would be nice to have explicit constructions. Establishing the exact size of deletion  balls  is a long-standing open question with many implications. On the practical side, there may be interesting  applications of covering codes based on insertions and deletions.  Finally, it would be worthwhile to extend  these results to edit distance, in which insertions, deletions, and substitutions are considered.

\appendix
\section{Omitted Calculations and Approximations}

\begin{proposition}\label{prop:insertion:approx:lb}
	For fixed $R$ and large $n$, it holds that 
\[\frac{q^{\dim+R}}{\VI^q(n,R)}
\geq 
\frac{R!q^{\dim+R}}{n^R(q-1)^R}(1-o(1)).\]
\end{proposition}
\begin{proof}
\begin{align*}
&\frac{q^{\dim+R}}{\sum_{i=0}^{R} \binom{\dim+R}{i}(q-1)^i}& \\
&= \frac{q^{\dim+R}}{\binom{n+R}{R}(q-1)^R\left(1+\sum_{i=0}^{R-1} \frac{\binom{n+R}{i}(q-1)^{i-R}}{\binom{n+R}{R}}\right) } & \\
& \overset{(a)}{\geq} \frac{q^{\dim+R}}{\binom{n+R}{R}(q-1)^R} \left(1-\sum_{i=0}^{R-1} \frac{\binom{n+R}{i}(q-1)^{i-R}}{\binom{n+R}{R}}\right) &\\
&\overset{(b)}{\geq} \frac{q^{\dim+R}}{\binom{n+R}{R}(q-1)^R} \left(1-R^R\sum_{i=0}^{R-1} \frac{(n+R)^i(q-1)^{i-R}}{(n+R)^R}\right) & \\
& \overset{(c)}{\geq} \frac{q^{\dim+R}}{\binom{n+R}{R}(q-1)^R} \left(1-\frac{R^{R+1}(q-1)^{-1}}{n+R}\right) & \\
&\overset{(d)}{\geq} \frac{R!q^{\dim+R}}{(n+R)^R(q-1)^R} \left(1-\frac{R^{R+1}(q-1)^{-1}}{n+R}\right) & \\
&= \frac{R!q^{\dim+R}}{n^R(1+\frac{R}{n})^R(q-1)^R} \left(1-\frac{R^{R+1}(q-1)^{-1}}{n+R}\right) & \\
&\overset{(e)}{\geq} \frac{R!q^{\dim+R}}{n^R(q-1)^R} \left(1 - \frac{R^2}{n}\right) \left(1-\frac{R^{R+1}(q-1)^{-1}}{n+R}\right) & \\
&= \frac{R!q^{\dim+R}}{n^R(q-1)^R}(1-o(1)),& 
\end{align*}
where we used in $(a),(e)$ the inequality $(1+x)^r \geq 1+rx$ for any $x>-1$ and any $r\leq0$ or $r\geq1$. In inequalities $(b),(d)$ we used that $n^R/R^R\leq \binom{n}{R} \leq n^R/R!$. Further, in $(c)$ we used that the largest of the terms in the sum is $i=R-1$ to bound the sum. Finally, the statement holds for fixed $R$ and large $n$.
%
%

\end{proof}

\begin{proposition}\label{prop:deletion:approx:lb}
	For fixed $R$ and large $n$, we have the asymptotic relation
	\[q\sum_{r=1}^{\dim-R} \frac{(q-1)^{r-1}\binom{\dim-R-1}{r-1}}{\binom{r+3R-1}{R}} \geq \frac{R!q^{n}}{n^R(q-1)^{R}}  \left( 1 -o(1)  \right).\]		
\end{proposition}
\begin{proof}
	We first note that {$$\binom{n-R-1}{r-1} = \frac{r(r+1)\cdots (r+R-1)}{(n-1)(n-2)\cdots (n-R)} \binom{n-1}{r+R-1}.$$}
	Hence,
	\begin{align*}
	& q\sum_{r=1}^{\dim-R} \frac{(q-1)^{r-1}\binom{\dim-R-1}{r-1}}{\binom{r+3R-1}{R}} 
	 = \frac{q}{(n-1)(n-2)\cdots (n-R)} & \\
	& \quad\cdot \sum_{r=1}^{\dim-R} \frac{(q-1)^{r-1}r(r+1)\cdots (r+R-1)}{\binom{r+3R-1}{R}}\binom{\dim-1}{r+R-1} & \\
	& \geq \frac{qR!}{n^R} \sum_{r=1}^{\dim-R} \frac{(q-1)^{r-1}r(r+1)\cdots (r+R-1)}{(r+2R)\cdots (r+3R-1)}\binom{\dim-1}{r+R-1}. &
	\end{align*}
	For $0\leq i\leq R-1$, it holds that $\frac{r+i}{r+2R+i} \geq\frac{r}{r+2R}$ and so
	\begin{align*}
	\frac{r(r+1)\cdots (r+R-1)}{(r+2R)\cdots (r+3R-1)} & \geq \left(\frac{r}{r+2R}\right)^R & \\
	& = \left(1-\frac{2R}{r+2R}\right)^R & \\
	& \geq 1- \frac{2R^2}{r+2R},& 
	\end{align*}
	where the last step holds by the 
	the inequality $(1-x)^d\geq 1-xd $ for $d\geq 1$. Furthermore, for $r\geq 2R(\sqrt{\dim}R-1) \overset{\mathrm{def}}{=} b$, we have that
	$1- \frac{2R^2}{r+2R} \geq 1-\frac{1}{\sqrt{\dim}}$, and thus we deduce that
	\begin{align*}
		& q\sum_{r=1}^{\dim-R} \frac{(q-1)^{r-1}\binom{\dim-R-1}{r-1}}{\binom{r+3R-1}{R}}  & \\
		&  \geq q(1-\frac{1}{\sqrt{\dim}})\frac{R!}{n^R} \sum_{r=b}^{\dim-R} (q-1)^{r-1}\binom{\dim-1}{r+R-1} & \\
		& = q(1-\frac{1}{\sqrt{\dim}})\frac{R!}{n^R} \sum_{r=b+R-1}^{\dim-1}(q-1)^{r-R} \binom{\dim-1}{r} & \\
		& = q(1-\frac{1}{\sqrt{\dim}})\frac{R!}{n^R(q-1)^{R}} \sum_{r=b+R-1}^{\dim-1}(q-1)^{r} \binom{\dim-1}{r} & \\
		& =(1-\frac{1}{\sqrt{\dim}}) \frac{R!}{n^R(q-1)^{R}} \left( q^{\dim} - \sum_{r=1}^{b+R-2} \binom{\dim-1}{r} \right) & \\
		& \geq \frac{R!q^{n}}{n^R(q-1)^{R}}  \left( 1 -o(1)  \right). & 
	\end{align*}
\end{proof}

\bibliographystyle{IEEEtran}
\bibliography{refs}

\end{document}